\documentclass[runningheads]{llncs}
\usepackage{graphicx}
\usepackage{stmaryrd}
\usepackage{amsmath}
\usepackage{amsfonts}
\usepackage[subrefformat=parens]{subcaption}
\newcommand{\func}[3]{#1 \colon #2 \to #3}
\newcommand{\bigO}[1]{\mathrm{O}(#1)}

\newcommand{\set}[1]{\left\{ #1 \right\}}
\newcommand{\multiset}[1]{\left\llbracket #1 \right\rrbracket}
\newcommand{\setST}[2]{\left\{ #1 \;\middle|\; #2 \right\}}
\newcommand{\multisetST}[2]{\left\llbracket #1 \mid #2 \right\rrbracket}
\newcommand{\sequence}[1]{\left\langle #1 \right\rangle}

\newcommand{\graph}[2]{\sequence{#1, #2}}
\newcommand{\edge}[2]{\set{#1, #2}}
\newcommand{\degree}[2]{\mathrm{deg}(#1; #2)} 
\newcommand{\degreeSet}[1]{D(#1)}

\newcommand{\degreeOfColor}[3]{\mathrm{deg}(#1; #2, #3)} 
\newcommand{\coloredDegree}[2]{\mathrm{deg}(#1; #2)}
\newcommand{\coloredDegreeSet}[1]{D(#1)}

\newcommand{\UnigraphNumber}[1]{w(#1)}
\newcommand{\strongUnigraphNumber}[1]{s(#1)}
\newcommand{\vertexCoverNumber}[1]{\tau(#1)}

\begin{document}
\title{Decomposing a Graph into Unigraphs}
%
%

\author{
Takashi Horiyama\inst{1} \and
Jun Kawahara\inst{2} \and
Shin-ichi Minato\inst{3} \and
Yu Nakahata\inst{3}
}
\authorrunning{T, Horiyama et al.}
%
\institute{
Saitama University, Saitama, Japan\\
\email{horiyama@al.ics.saitama-u.ac.jp}\and
Nara Institute of Science and Technology, Nara, Japan\\
\email{jkawahara@is.naist.jp}\and
Kyoto University, Kyoto, Japan\\
\email{\{minato@i, nakahata.yu.27e@st\}.kyoto-u.ac.jp}
}

\maketitle              
\begin{abstract}
Unigraphs are graphs uniquely determined by their own degree sequence up to isomorphism.
There are many subclasses of unigraphs such as threshold graphs, split matrogenic graphs, matroidal graphs, and matrogenic graphs.
Unigraphs and these subclasses are well studied in the literature.
Nevertheless, there are few results on superclasses of unigraphs.
In this paper, we introduce two types of generalizations of unigraphs: \emph{$k$-unigraphs} and \emph{$k$-strong unigraphs}.
We say that a graph $G$ is a \emph{$k$-unigraph} if $G$ can be partitioned into $k$ unigraphs.
$G$ is a \emph{$k$-strong unigraph} if not only each subgraph is a unigraph but also the whole graph can be uniquely determined up to isomorphism, by using the degree sequences of all the subgraphs in the partition.
We describe a relation between $k$-strong unigraphs and the subgraph isomorphism problem.
We show some properties of $k$-(strong) unigraphs and algorithmic results on calculating the minimum $k$ such that a graph $G$ is a $k$-(strong) unigraph.
This paper will open many other research topics.
\end{abstract}

\section{Introduction}\label{sec:intro}
Unigraphs~\cite{johnson1975simple,li1975graphic} are graphs uniquely determined by their own degree sequence up to isomorphism.
The family of unigraphs contains many important graph classes such as threshold graphs~\cite{chvatal1977aggregation}, split matrogenic graphs~\cite{hammer2004splitoids}, matroidal graphs~\cite{peled1977matroidal}, and matrogenic graphs~\cite{foldes1978class}.
The relationship of inclusion among these classes is as follows:
\begin{equation}
    \mbox{threshold} \subset \mbox{split matrogenic} \subset \mbox{matroidal} \subset \mbox{matrogenic} \subset \mbox{unigraph}.
\end{equation}
Unigraphs and these subclasses are well studied in the literature.
Nevertheless, there are few results on superclasses of unigraphs.
Although there are many approaches to consider superclasses of some graph class, in this paper, we pay attention to a superclass obtained by generalizing a graph class by the partition of the edge set of a graph.
For example, the \emph{arboricity} of a graph is the minimum number of forests into which the edge set of the graph can be partitioned.
The \emph{thickness} of a graph is similarly defined by the minimum number of planar subgraphs into which its edges can be partitioned.
These values of a graph are the measures of how far the graph is from the original graph classes.

\begin{figure}[t]
    \centering
    \begin{subfigure}{.45\linewidth}
        \centering
        \includegraphics[bb=0 0 90 53, scale=0.5]{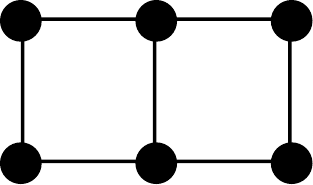}
        \caption{Domino graph.}
        \label{fig:domino}
    \end{subfigure}
    \hfill
    \begin{subfigure}{.45\linewidth}
        \centering
        \includegraphics[bb=0 0 223 97, scale=0.5]{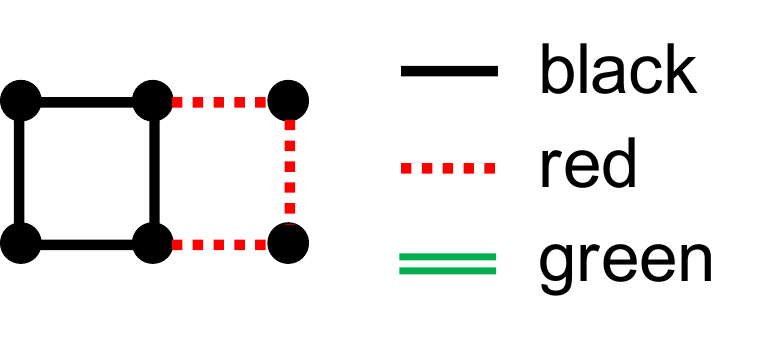}
        \caption{2-coloring of the domino graph. Each colored subgraph forms a unigraph.}
        \label{fig:domino_2color}
    \end{subfigure}
    \hfill
    \begin{subfigure}{.45\linewidth}
        \centering
        \includegraphics[bb=0 0 90 75, scale=0.5]{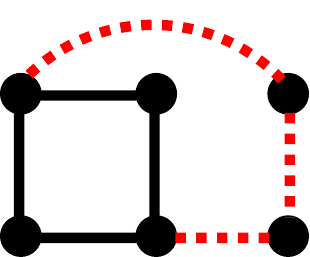}
        \caption{2-colored graph whose colored degree sequence is the same as the edge-colored graph in Figure~\ref{fig:domino_2color}.}
        \label{fig:domino_2color_not_iso}
    \end{subfigure}
    \hfill
    \begin{subfigure}{.45\linewidth}
        \centering
        \includegraphics[bb=0 0 89 53, scale=0.5]{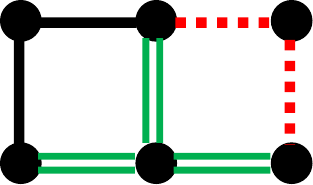}
        \caption{3-colored graph whose colored degree sequence determines the whole graph uniquely up to isomorphism.}
        \label{fig:domino_3color}
    \end{subfigure}
    \caption{Domino graph and its edge colorings.}
    \label{fig:domino_and_coloring}
\end{figure}

In this paper, we introduce two types of generalizations of unigraphs: \emph{$k$-unigraphs} and \emph{$k$-strong unigraphs}.
We say that a graph is a \emph{$k$-unigraph} if its edge set can be partitioned into $k$ sets such that the subgraph induced by each set is a unigraph.
We call the minimum $k$ satisfying the above condition the \emph{unigraph number} of the graph.
The formal definition is described in Section~\ref{sec:unigraphs}.
The unigraph number of a graph is a measure of how far the graph is from a unigraph.
For example, let us consider the domino graph, which is shown in Figure~\ref{fig:domino}.
The domino graph is not a unigraph because the graph in Figure~\ref{fig:domino_not_iso} has the same degree sequence $(3, 3, 2, 2, 2, 2)$ as the domino graph but is not isomorphic to the domino graph.
However, we can decompose the domino graph into two unigraphs as shown in Figure~\ref{fig:domino_2color}.
The black and the red subgraph have the degree sequences $(2, 2, 2, 2)$ and $(2, 2, 1, 1)$, respectively.
Both the subgraphs are unigraphs.
Therefore, the domino graph is a $2$-unigraph and its unigraph number is two.

However, the definition of $k$-unigraphs is somewhat \emph{weak}, that is, while a unigraph requires that its degree sequence determines the \emph{whole} graph uniquely, in a $k$-unigraph, only \emph{each subgraph} in a partition has to be uniquely determined by the degree sequence of the subgraph.
Therefore, we introduce another generalization of unigraphs: \emph{$k$-strong unigraphs}.
We say that a graph is a $k$-strong unigraph if not only each subgraph is a unigraph but also the whole graph can be uniquely determined up to isomorphism, by using the degree sequences of all the subgraphs in the partition.
Let us consider this in a more detailed way.
A partition of an edge set can be seen as an edge coloring.
A \emph{$k$-edge coloring} of a graph is an assignment of $k$ colors to the edges of the graph.
Note that $k$-edge colorings do not have to be proper, that is, it is allowed that two edges sharing a vertex are assigned the same color.
When we are given a graph and its $k$-edge coloring, a \emph{colored degree} of a vertex is a $k$-tuple, whose $j$-th element represents the number of edges in the $j$-th color.
We define the \emph{colored degree sequence} by the sequence of the colored degrees of all the vertices.
We say that a graph is a \emph{$k$-strong unigraph} if there exists a $k$-edge coloring such that not only each colored subgraph is a unigraph but also the colored degree sequence obtained by the coloring determines the whole graph uniquely up to isomorphism.
We call the minimum $k$ satisfying the above condition the \emph{strong unigraph number} of the graph.
The strong unigraph number of the graph is the measure of not only how far the graph is from a unigraph, but also how difficult it is to determine the graph uniquely up to isomorphism by degree sequences.
For example, the two edge-colored graphs in Figures~\ref{fig:domino_2color} and \ref{fig:domino_2color_not_iso} have the same colored degree sequence $((2, 1), (2, 1), (2, 0), (2, 0), (0, 2), (0, 2))$, where, for each tuple, the first and the second element mean the numbers of the black and the red edges incident to a vertex, respectively.
Therefore, the edge colorings do not determine the whole graphs uniquely.
On the other hand, Figure~\ref{fig:domino_3color} shows a 3-edge coloring of the domino graph.
The colored degree sequence obtained by the edge coloring is $((2, 0, 0), (1, 1, 1), (1, 0, 1), (0, 2, 0), (0, 1, 1), (0, 0, 3))$, where, for each tuple, the first, the second, and the third element mean the numbers of the black, the red, and the green edges, respectively.
In the edge coloring, not only each colored subgraph is a unigraph but also the whole graph is uniquely determined by the colored degree sequence.
Therefore, the domino graph is a $3$-strong unigraph.

\begin{figure}[t]
    \centering
    \includegraphics[bb=0 0 127 54, scale=0.5]{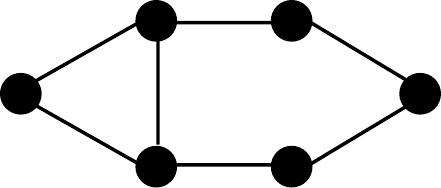}
    \caption{Graph whose degree sequence is the same as the domino graph but is not isomorphic to the domino graph.}
    \label{fig:domino_not_iso}
\end{figure}

The strong unigraph number is related to the well-known \emph{subgraph isomorphism problem}.
Especially, we focus on its variant, the \emph{isomorphic subgraph enumeration problem}.
In the problem, the input is a pair of graphs: a host graph and a query graph.
The task is to find all the subgraphs, of the host graph, such that they are isomorphic to the query graph.
Since the subgraph isomorphism problem is in NP-complete, enumerating isomorphic subgraphs is a hard task in general.
One approach to the problem is a technique using decision diagrams~\cite{kawahara2019colorful}.
In this approach, one first finds an edge coloring of the query graph such that the colored degree sequence obtained by the coloring determines the whole graph uniquely.
Then one searches for subgraphs, of the host graph, which yields the same colored degree sequence as the edge coloring of the query graph.
Using edge colorings, we only have to maintain colored degree sequences of subgraphs.
It leads to an efficient algorithm.
A notable point is that the complexity of the algorithm heavily depends on the number of colors used in the edge coloring.
Therefore, it is important to find an edge coloring that not only determines the query graph uniquely but also uses as few colors as possible.
However, in their study, they only gave edge colorings for a few special query graphs in an ad hoc way.
In order to deal with various query graphs efficiently, a theoretical investigation is essential.
This motivates us to define $k$-strong unigraphs.
We focus on the fact that graphs which need only one color to determine the whole graph uniquely are exactly unigraphs.
Then we generalize them by the number of necessary colors and define $k$-strong unigraphs.
The reason why we define $k$-unigraphs, along with $k$-strong unigraphs, is that they seem to be easier to be dealt with than $k$-strong unigraphs and they are of theoretical interest themselves.
In addition, the unigraph number of a graph can be used as a lower bound of the strong unigraph number.

Our contributions in this paper are:
\begin{itemize}
    \item We introduce two types of generalizations of unigraphs: \emph{$k$-unigraphs} and \emph{$k$-strong unigraphs}.
    \item We analyze a basic property of the (strong) unigraph number of a graph. We show that, for every graph, its (strong) unigraph number is at most its \emph{vertex cover number}.
    \item We show that, for every tree, its unigraph number is equal to its \emph{edge domination number}, the minimum number of edges such that every edge shares at least one endpoint with an adopted edge. Using this property, we show that we can calculate the unigraph number of a tree in linear time.
\end{itemize}

This paper is organized as follows.
The rest of this section describes related work.
Section~\ref{sec:preliminaries} gives preliminaries.
In Section~\ref{sec:unigraphs}, we introduce new generalizations of unigraphs: $k$-unigraphs and $k$-strong unigraphs.
Section~\ref{sec:star} presents a relationship between the (strong) unigraph number and the vertex cover number.
Section~\ref{sec:tree} shows the relationship between the unigraph number of a tree and the edge domination number and that we can calculate the unigraph number of a tree in linear time.
Section~\ref{sec:conclusions} gives concluding remarks.

\subsection*{Related work}\label{sec:related_work}
Unigraphs~\cite{johnson1975simple,li1975graphic} and its subclasses~\cite{chvatal1977aggregation,foldes1978class,hammer2004splitoids,peled1977matroidal} are well studied in the literature.
There is a linear-time algorithm to recognize unigraphs~\cite{kleitman1975note}.
Tyshkevich~\cite{tyshkevich2000decomposition} characterized unigraphs by the canonical decomposition\footnote{Tyshkevich decomposed a unigraph into several graphs. On the other hand, we decompose a graph (not necessarily a unigraph) into several unigraphs.}.
Borri et al.~\cite{borri2011recognition} presented a new linear-time algorithm for recognizing unigraphs by extending the pruning algorithm for recognizing matrogenic graphs~\cite{marchioro1984degree}.
Unigraphs are not hereditary, that is, an induced subgraph of a unigraph may not be a unigraph.
\emph{Hereditary unigraphs}, a subclass of unigraphs that are hereditary, are studied by Barrus~\cite{barrus2012switch,barrus2013hereditary}.

A subgraph enumeration technique using decision diagrams is well studied.
It has been used for many applications such as network reliability evaluation~\cite{hardy2007kterminal}, electrical distribution network~\cite{inoue2014distribution}, influence spread~\cite{maehara2017exact}, and so on.
A framework of algorithms to enumerate subgraphs using decision diagrams is called frontier-based search~\cite{sekine1995tutte}.
Kawahara et al.~\cite{kawahara2019colorful} have proposed generalized frontier-based search, colorful frontier-based search.
By their framework, many types of query graphs can be dealt with in a unified way.
However, the efficiency of the algorithm heavily depends on how to color the query graph.
There are few theoretical results on how many colors are needed to determine the query graph uniquely, which motivates us to define $k$-unigraphs and $k$-strong unigraphs.

\section{Preliminaries}\label{sec:preliminaries}
Let $G = \graph{V(G)}{E(G)}$ be a graph with $n = |V(G)|$ vertices and $m = |E(G)|$ edges.
We assume that $G$ is finite, undirected, connected, and simple (without multiple edges and loops).
We regard an edge $e$ as a 2-vertex set $\edge{u}{v}$, where $u$ and $v$ are the two distinct endpoints of $e$.
For a vertex $v \in V(G)$, the \emph{degree of} $v$ is $\degree{G}{v} = |\setST{e \in E(G)}{v \in e}|$.
For $U \subseteq V(G)$, the \emph{vertex-induced subgraph} by $U$ is $G[U] = \graph{U}{E'}$, where $E' = \setST{e \in E(G)}{e \subseteq U}$.
For $F \subseteq E(G)$, the \emph{edge-induced subgraph} by $F$ is $G[F] = \graph{V'}{F}$, where $V' = \bigcup_{e \in F} e$.
For two graphs $G$ and $H$, $G$ is \emph{isomorphic} to $H$ if there exists a bijection $\func{f}{V(G)}{V(H)}$ such that $\edge{u}{v} \in E(G)$ if and only if $\edge{f(u)}{f(v)} \in E(H)$.
When $G$ is isomorphic to $H$, we write $G \simeq H$ (or $H \simeq G$, since the isomorphism relation is symmetric).
For a positive integer $k$, we define $[k] = \set{1, 2, \dots, k}$.

For convenience, we define a \emph{degree set} instead of a degree sequence.
The \emph{degree set of $G$} is the multiset $\degreeSet{G} = \multisetST{\degree{G}{v}}{v \in V(G)}$\footnote{We use "$\multiset{}$" for multisets instead of "$\set{}$".}.
Using degree sets, we can define unigraphs as follows.

\begin{definition}[unigraph~\cite{johnson1975simple,li1975graphic}]
A graph $G$ is a \emph{unigraph} if, for all graphs $H$ whose degree sets are the same as $G$, $H$ is isomorphic to $G$.
\end{definition}

The following lemma shows there are few patterns for disconnected unigraphs.
Although the lemma is basic, an explicit statement could not be found in the literature.
Thus, we proof the lemma in Appendix~\ref{sec:proof_disconnected_unigraph}.

\begin{lemma}\label{le:disconnected_unigraph}
    If $G$ is a disconnected unigraph, $G$ has at most one connected component which has at least three vertices.
\end{lemma}

By Lemma~\ref{le:disconnected_unigraph}, we restrict our attention to decomposing a graph into connected unigraphs in this paper.
The following lemma and corollary are basic characterizations of (connected) unigraphs.

\begin{lemma}[\cite{barrus2012switch}]
    A unigraph does not contain a path with four edges as a vertex-induced subgraph.
\end{lemma}
\begin{corollary}\label{cor:diameter}
    If $G$ is a connected unigraph, its diameter is at most three.
\end{corollary}

\section{Two generalizations of unigraphs}\label{sec:unigraphs}
In this paper, we introduce two types of generalizations of unigraphs: \emph{$k$-unigraphs} and \emph{$k$-strong unigraphs}.
Roughly speaking, we say that a graph $G$ is a \emph{$k$-unigraph} if $G$ can be partitioned into $k$ unigraphs.
$G$ is a \emph{$k$-strong unigraph} if not only each subgraph is a unigraph but also the whole graph can be uniquely determined up to isomorphism, by using the degree sequences of all the subgraphs.

In order to formally define $k$-unigraphs and $k$-strong unigraphs, we introduce \emph{$k$-edge colorings} and \emph{colored degrees}.
A \emph{$k$-edge coloring} of $G$ is a function $\func{c}{E(G)}{[k]}$.
It assigns each edge in $G$ a color $i \in [k]$.
Note that $k$-edge colorings do not have to be proper, that is, it is allowed that two edges sharing a vertex are assigned the same color.
The pair of a graph $G$ and an edge coloring $c$ is an \emph{edge-colored graph}, which we denote by $G^c$.
If $c(e) = i \in [k]$ holds for an edge $e$, we say that \emph{the color of $e$ is $i$} or that \emph{$e$ has the color $i$}.
For each color $i \in [k]$, let $E_i(G^c)$ be the set of edges in the color $i$, that is, $E_i(G^c) = \setST{e \in E(G)}{c(e) = i}$.
We call $G^c_i = G[E_i(G^c)]$ the \emph{color-$i$ subgraph of $G^c$}.
For a color $i \in [k]$ and a vertex $v \in V(G)$, we define the \emph{color-$i$ degree of $v$} by the number of the color-$i$ edges incident to $v$ and denote it by $\degreeOfColor{G^c}{v}{i} = |\setST{e \in E_i(G)}{v \in e}|$.
We define the \emph{colored degree of $v \in V(G)$} by $\coloredDegree{G^c}{v} = \sequence{\degreeOfColor{G^c}{v}{1}, \degreeOfColor{G^c}{v}{2}, \dots, \degreeOfColor{G^c}{v}{k}}$.
Finally, we define the \emph{colored degree set of $G^c$} by the multiset  $\coloredDegreeSet{G^c} = \multisetST{\coloredDegree{G^c}{v}}{v \in V(G)}$.

Now we define \emph{$k$-(strong) unigraphs} using \emph{$k$-(strongly) unigraphic (edge) colorings}, which we define in the following.
By Lemma~\ref{le:disconnected_unigraph}, which states there are few patterns for disconnected unigraphs, we restrict our attention to decomposing a graph into connected unigraphs.

\begin{definition}[$k$-(strongly) unigraphic coloring]
    Let $\func{c}{E(G)}{[k]}$ be a $k$-edge coloring.
    \begin{itemize}
        \item The $k$-edge coloring $c$ is a \emph{$k$-unigraphic (edge) coloring} if, for all $i \in [k]$, the color-$i$ subgraph of $G^c$ is a connected unigraph.
        \item The $k$-edge coloring $c$ is a \emph{$k$-strongly unigraphic (edge) coloring} if it is a $k$-unigraphic coloring and, for all graphs $H = \graph{V(H)}{E(H)}$, the following holds:
        \begin{equation}
            \exists \func{c'}{E(H)}{[k]}, \coloredDegreeSet{G^c} = \coloredDegreeSet{H^{c'}} \Rightarrow G \simeq H.
        \end{equation}
    \end{itemize}
\end{definition}

\begin{definition}[$k$-(strong) unigraph]
    $G$ is a \emph{$k$-(strong) unigraph} if there exists a $k$-(strongly) unigraphic coloring of $G$.
\end{definition}

We define the \emph{unigraph number of $G$} as the minimum $k$ such that $G$ is a $k$-unigraph and denote it by $w(G)$.
We also define the \emph{strong unigraph number} of $G$ similarly and denote it by $s(G)$.

\section{Upper bound of the (strong) unigraph numbers for general graphs}\label{sec:star}
In this section, we show that, for every graph, both its unigraph number and strong unigraph number are at most its \emph{vertex cover number}.
A set $U \subseteq V(G)$ is called a \emph{vertex cover} of $G$ if, for all $e \in E(G)$, $e \cap U \neq \emptyset$ holds.
A vertex cover of $G$ with the minimum size is called a \emph{minimum vertex cover} and its size is the vertex cover number, which we denote by $\vertexCoverNumber{G}$.

To prove the bound, we utilize a \emph{$k$-star coloring}, a $k$-edge coloring whose each colored subgraph forms a star.
We show that a $k$-star coloring is a $k$-(strongly) unigraphic coloring and that we can construct a $\vertexCoverNumber{G}$-star coloring for any graph $G$, which means that both $\UnigraphNumber{G}$ and $\strongUnigraphNumber{G}$ are at most $\vertexCoverNumber{G}$.

\begin{definition}[star coloring]
    A $k$-edge coloring $\func{c}{E(G)}{[k]}$ is a \emph{$k$-star coloring} if, for all $i \in [k]$, the color-$i$ subgraph of $G^c$ is a star.
\end{definition}
\begin{lemma}\label{le:star_weak}
    Every $k$-star coloring is a $k$-unigraphic coloring.
\end{lemma}
\begin{proof}
    This is because a star is a unigraph. \qed
\end{proof}
\begin{corollary}\label{co:star}
    Let $\func{c}{E(G)}{[k]}$ be a $k$-star coloring of $G$.
    For all graphs $H$ and a $k$-edge coloring $\func{c'}{E(H)}{[k]}$, $c'$ is a $k$-star coloring if $D(G^c) = D(H^{c'})$.
\end{corollary}

\begin{figure}[t]
    \centering
    \begin{subfigure}{0.45\linewidth}
        \centering
        \includegraphics[bb=0 0 146 126, scale=0.5]{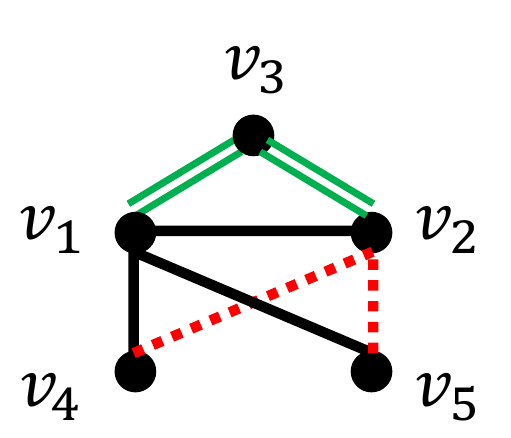}
        \caption{3-star coloring whose stars have centers $v_1, v_2$, and $v_3$.}
        \label{fig:star_coloring}
    \end{subfigure}
    \hfill
    \begin{subfigure}{0.45\linewidth}
        \centering
        \includegraphics[bb=0 0 147 126, scale=0.5]{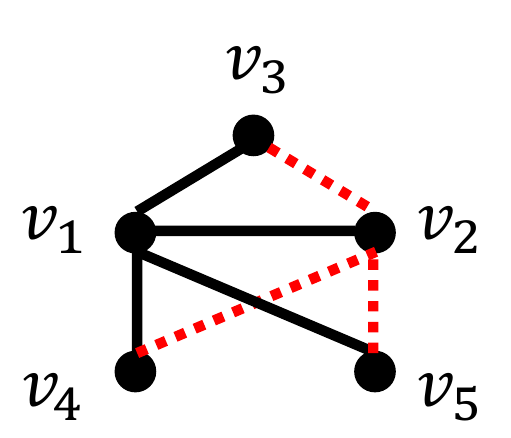}
        \caption{2-star coloring whose stars have centers $v_1$ and $v_2$. The two vertices form a minimum vertex cover.}
        \label{fig:star_coloring_vertex_cover}
    \end{subfigure}
    \caption{Star colorings with different numbers of colors for the same graph.}
    \label{fig:star_colorings}
\end{figure}

Figure~\ref{fig:star_colorings} shows examples of star colorings.
In a star coloring, roughly speaking, whether two vertices are adjacent depends only on their colored degrees.
Formally, this can be stated as the following lemma.

\begin{lemma}\label{le:star}
    Let $\func{c}{E(G)}{[k]}$ be a $k$-star coloring.
    $\edge{u}{v} \in E(G)$ if and only if there exists a color $i \in [k]$ such that
    \begin{enumerate}
        \item $|E_i(G^c)| = 1$ and $\degreeOfColor{G^c}{u}{i} = \degreeOfColor{G^c}{v}{i} = 1$, or
        \item $|E_i(G^c)| \geq 2$, $\min_{x \in \{u, v\}} \degreeOfColor{G^c}{x}{i} = 1$ and $\max_{x \in \{u, v\}} \degreeOfColor{G^c}{x}{i} \geq 2$.
    \end{enumerate}
\end{lemma}
\begin{proof}
    First, we show the ``only if'' part.
    Let $\func{c}{E(G)}{[k]}$ be a $k$-star coloring.
    If $\edge{u}{v} \in E(G)$, let $i = c(\{u, v\})$.
    If $|E_i(G^c)| = 1$, $G^c_i$ is isomorphic to $K_{2}$. In this case, both $u$ and $v$ are the endpoints of the only edge of the $K_2$. Therefore, $\degreeOfColor{G^c}{u}{i} = \degreeOfColor{G^c}{v}{i} = 1$ holds.
    If $|E_i(G^c)| \geq 2$, since $c$ is a star coloring, $G^c_i$ is isomorphic to the star $K_{1, p}$ for some integer $p \geq 2$.
    In this case, one of $u$ and $v$ is the center of the star, and the other is a leaf of the star.
    Therefore, $\min_{x \in \{u, v\}} \degreeOfColor{G^c}{x}{i} = 1$ and $\max_{x \in \{u, v\}} \degreeOfColor{G^c}{x}{i} \geq 2$ hold.

    Next, we show the ``if'' part.
    Assume that there exists a color $i \in [k]$ such that $|E_i(G^c)| = 1$ and $\degreeOfColor{G^c}{u}{i} = \degreeOfColor{G^c}{v}{i} = 1$.
    $|E_i(G^c)| = 1$ implies that $G^c_i \simeq K_2$.
    In addition, by $\degreeOfColor{G^c}{u}{i} = \degreeOfColor{G^c}{v}{i} = 1$, both $u$ and $v$ are the endpoints of the $K_2$. Therefore, $\edge{u}{v} \in E(G)$ holds.
    Assume that there exists a color $i \in [k]$ such that $|E_i(G^c)| \geq 2$, $\min_{x \in \{u, v\}} \degreeOfColor{G^c}{x}{i} = 1$ and $\max_{x \in \{u, v\}} \degreeOfColor{G^c}{x}{i} \geq 2$.
    Since $c$ is a star coloring, $|E_i(G^c)| \geq 2$ implies that $G^c_i \simeq K_{1, p}$ for some integer $p \geq 2$.
    Therefore, $\min_{x \in \{u, v\}} \degreeOfColor{G^c}{x}{i} = 1$ and $\max_{x \in \{u, v\}} \degreeOfColor{G^c}{x}{i} \geq 2$ implies that one of $u$ and $v$ is the center of the star and the other is a leaf, meaning that $\edge{u}{v} \in E(G)$. \qed
\end{proof}

\begin{lemma}\label{le:star_strong}
    Every $k$-star coloring is a $k$-strongly unigraphic coloring.
\end{lemma}
\begin{proof}
    Let $\func{c}{E(G)}{[k]}$ be a $k$-star coloring of $G$,
    $H$ be a graph, and $\func{c'}{E(H)}{[k]}$ be a $k$-edge coloring of $H$ such that $D(G^c) = D(H^{c'})$.
    We show that $G \simeq H$, that is, there exists a bijection $\func{g}{V(G)}{V(H)}$ such that, for all $u, v \in V(G)$,  $\edge{u}{v} \in E(G)$ if and only if $\edge{g(u)}{g(v)} \in E(H)$.
    Consider a bijection which preserves the colored degrees of the vertices, that is, $\func{f}{V(G)}{V(H)}$ such that $f(a) = b \Rightarrow \coloredDegree{G^c}{a} = \coloredDegree{H^{c'}}{b}$, where $a \in V(G)$ and $b \in V(H)$.
    There exists such a bijection because $\coloredDegreeSet{G^c} = \coloredDegreeSet{H^{c'}}$.
    We show that, for all $u, v \in V(G)$, $\edge{u}{v} \in E(G) \Leftrightarrow \edge{f(u)}{f(v)} \in E(H)$ holds.

    Let $u, v \in V(G)$.
    If $\edge{u}{v} \in E(G)$, let $i = c(\edge{u}{v})$.
    We consider two cases based on Lemma~\ref{le:star}.
    If $|E_i(G^c)| = 1$, $\degreeOfColor{G^c}{u}{i} = \degreeOfColor{G^c}{v}{i} = 1$ holds.
    Since $f$ preserves the colored degrees of the vertices, $\degreeOfColor{H^{c'}}{f(u)}{i} = \degreeOfColor{H^{c'}}{f(v)}{i} = 1$ holds.
    In addition, $|E_i(H^{c'})| = 1$ because $D(G^c) = D(H^{c'})$.
    Moreover, by Corollary~\ref{co:star}, $c'$ is a $k$-star coloring.
    Therefore, by the first case of Lemma~\ref{le:star}, $\edge{f(u)}{f(v)} \in E(H)$ holds. Let us consider the second case, where $\edge{u}{v} \in E(G)$ and $|E_i(G^c)| \geq 2$.
    From the second case of Lemma~\ref{co:star}, without loss of generality, we assume that $\degreeOfColor{G^c}{u}{i} \geq 2$ and $\degreeOfColor{G^c}{v}{i} = 1$.
    Since $f$ preserves the colored degrees of the vertices, $\degreeOfColor{H^{c'}}{f(u)}{i} \geq 2$ and $\degreeOfColor{H^{c'}}{f(v)}{i} = 1$ hold.
    In addition, $|E_i(H^{c'})| \geq 2$ and $c'$ is a $k$-star coloring.
    Therefore, by the second case of Lemma~\ref{le:star}, $\edge{f(u)}{f(v)} \in E(H)$ holds.

    If $\edge{u}{v} \notin E(G)$, since $c$ is a star coloring, neither the first nor the second conditions in Lemma~\ref{le:star} hold.
    Since $f$ preserves the colored degrees of the vertices, the same goes to $H^{c'}$.
    Therefore, by a similar discussion to the above, $\edge{u}{v} \notin E(H)$ holds.

    From the above discussion, $\edge{u}{v} \in E(G) \Leftrightarrow \edge{f(u)}{f(v)} \in E(H)$ holds. \qed
\end{proof}

Now we show that, for every graph $G$, the (strong) unigraph number is at most the vertex cover number, by constructing $\vertexCoverNumber{G}$-(strongly) unigraphic coloring of $G$.

\begin{theorem}
    $\UnigraphNumber{G} \leq \strongUnigraphNumber{G} \leq \vertexCoverNumber{G}$.
\end{theorem}
\begin{proof}
    $\UnigraphNumber{G} \leq \strongUnigraphNumber{G}$ follows from the definition.
    We show that $\strongUnigraphNumber{G} \leq \vertexCoverNumber{G}$ holds.
    Let $k$ be $\vertexCoverNumber{G}$ and
    $S = \set{v_1, \dots, v_k} \subseteq V(G)$ be a minimum vertex cover of $G$.
    We subscript the other vertices by $V \setminus S = \set{v_{k+1}, \dots, v_{n}}$.
    Let us color each edge $\edge{v_i}{v_j} \in E(G)$ by the color $\min \set{i, j}$.
    Observe that each edge is colored by some color in $[k]$ because $S$ is a vertex cover, that is, for each edge $\edge{v_i}{v_j}$, $v_i \in S$ or $v_j \in S$ holds.
    In addition, for each color $i \in [k]$, the color-$i$ subgraph forms a star by construction.
    Therefore, the obtained edge coloring is a $k$-star coloring.
    By Lemma~\ref{le:star_strong}, it is a $k$-strongly unigraphic coloring.
    This means that $\strongUnigraphNumber{G} \le \vertexCoverNumber{G}$ holds.
    An example of a star coloring which is obtained by a minimum vertex cover is shown in Figure~\ref{fig:star_coloring_vertex_cover}. \qed
\end{proof}

Note that the bound is not tight.
Complete graphs are the worst examples.
Since complete graphs are unigraphs, the unigraph numbers and the strong unigraph numbers are one regardless of their size, while the vertex cover number of the complete graph with $n$ vertices is $n-1$.
It is an interesting open problem whether we can obtain a tighter bound.

\section{Unigraph numbers for trees}\label{sec:tree}
\begin{figure}[t]
    \centering
    \begin{subfigure}{0.20\linewidth}
        \centering
        \includegraphics[bb=0 0 37 115, scale=0.4]{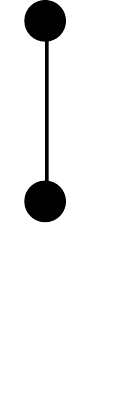}
        \caption{$K_2$.}
        \label{fig:K2}
    \end{subfigure}
    \begin{subfigure}{0.30\linewidth}
        \centering
        \includegraphics[bb=0 0 71 118, scale=0.4]{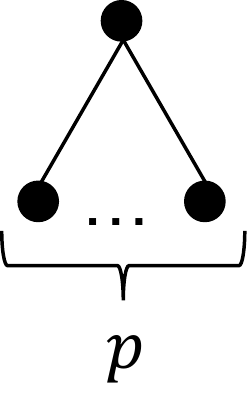}
        \caption{$K_{1, p}$, where $p \geq 2$.}
        \label{fig:star}
    \end{subfigure}
    \begin{subfigure}{0.45\linewidth}
        \centering
        \includegraphics[bb=0 0 168 118, scale=0.4]{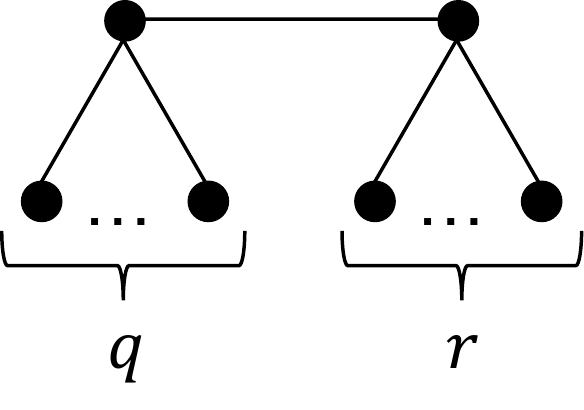}
        \caption{$S_{q, r}$, where $q \geq 1$ and $r \geq 1$.}
        \label{fig:bistar}
    \end{subfigure}
    \caption{Unigraphs in trees.}
    \label{fig:tree_unigraphs}
\end{figure}
In this section, we show that the unigraph number of a tree can be calculated in linear time.
Our proof consists of three steps.
First, we show that there are few patterns of unigraphs in a tree.
Second, we show that, for every tree, the unigraph number equals the \emph{edge domination number}.
An edge dominating set~\cite{yannakakis1980edge} of a graph $G$ is a set $E' \subseteq E(G)$ such that, for every edge $e \in E(G) \setminus E'$, there exists an edge $e' \in E'$ such that $e' \cap e \neq \emptyset$.
The edge domination number is the minimum size of edge dominating sets.
Third, using a linear-time algorithm to calculate the edge domination number of a tree~\cite{hedetniemi1977edge}, we calculate the unigraph number of the tree.

In the following, we use three patterns of graphs shown in Figure~\ref{fig:tree_unigraphs}: $K_2, K_{1, p}$, and $S_{q, r}$.
$K_2$ is the graph having only one edge (Figure~\ref{fig:K2}), $K_{1, p}\ (p \geq 2)$ is the star with $p$ leaves (Figure~\ref{fig:star}), and $S_{q, r}\ (q, r \geq 1)$ is the graph obtained by adding an edge between the two centers of $K_{1, q}$ and $K_{1, r}$ (Figure~\ref{fig:bistar}).

\begin{lemma}\label{le:tree_pattern}
    Let $G$ be a tree and $\func{c}{E(G)}{[k]}$ be a $k$-edge coloring of $G$.
    The edge coloring $c$ is a $k$-unigraphic coloring of $G$ if and only if, for each color $i \in [k]$, the color-$i$ subgraph is isomorphic to one of
    \begin{equation}
        K_2, K_{1, p}, S_{q, r},\label{eq:tree_unigraph}
    \end{equation}
    where $p \geq 2, q \geq 1,$ and $r \geq 1$ are positive integers.
\end{lemma}
\begin{proof}
    Since all the graphs of \eqref{eq:tree_unigraph} are unigraphs, the ``if'' part is correct.
    We show the ``only if'' part.
    If $c$ is a $k$-unigraphic coloring, by definition, for each color $i \in [k]$, the color-$i$ subgraph $G^c_i$ must be a unigraph.
    By Corollary~\ref{cor:diameter}, for each color $i \in [k]$, the diameter of $G^c_i$ must be at most three.
    Since $G^c_i$ is a subgraph of a tree $G$, $G^c_i$ is also a tree.
    A tree has the diameter at most three if and only if it is isomorphic to one of the graphs of \eqref{eq:tree_unigraph}. \qed
\end{proof}

\begin{lemma}\label{le:tree_dom}
    For every tree $G$, its unigraph number equals its edge domination number.
\end{lemma}
\begin{proof}
    Let $\func{c}{E(G)}{[k]}$ be a $k$-unigraphic coloring of $G$.
    We show that we can obtain an edge dominating set of size $k$ by choosing one edge from each color.
    By Lemma~\ref{le:tree_pattern}, for each color $i \in [k]$, the color-$i$ subgraph $G^c_i$ is isomorphic to one of the graphs of \eqref{eq:tree_unigraph}.
    If $G^c_i \simeq K_2$, we choose its only edge.
    If $G^c_i \simeq K_{1, p}$ for some integer $p \geq 2$, we choose an arbitrary edge of it.
    If $G^c_i \simeq S_{q, r}$ for some integers $q \geq 1$ and $r \geq 1$, we adopt the edge connecting the two centers of the stars $K_{1, q}$ and $K_{1, r}$.
    The set of the chosen edges is an edge dominating set.
    Figure~\ref{fig:tree_colorings} shows examples of a unigraphic coloring and an edge dominating set obtained by the above procedure.

    Let $E' = \set{e'_1, e'_2, \dots, e'_k} \subseteq E(G)$ be an edge dominating set with $k$ edges.
    We show that we can obtain a $k$-unigraphic coloring.
    Consider a $k$-edge coloring $\func{c}{E(G)}{[k]}$ which colors an edge $e'_i \in E'$ by the color $i$ and an edge $e \in E(G) \setminus E'$ by the minimum color $i \in [k]$ such that $e'_i \cap e \neq \emptyset$.
    Since $E'$ is an edge dominating set, all the edges are colored in this procedure.
    In this edge coloring, for each color $i \in [k]$, the color-$i$ subgraph is isomorphic to one of the graphs of \eqref{eq:tree_unigraph}.
    By Lemma~\ref{le:tree_pattern}, $c$ is a $k$-unigraphic coloring.
    Figure~\ref{fig:tree_edge_dominating_set} is an example of an edge dominating set and Figure~\ref{fig:tree_coloring} is a corresponding  unigrahic coloring. \qed
\end{proof}

\begin{figure}[t]
    \centering
    \begin{subfigure}{0.45\linewidth}
        \centering
        \includegraphics[bb=0 0 112 99, scale=0.6]{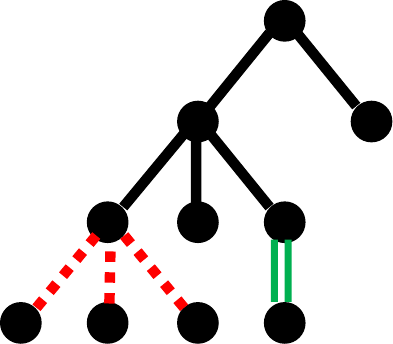}
        \caption{Unigraphic coloring. The black, the red, and the green subgraphs are isomorphic to $S_{3, 1}, K_{1, 3}$ and $K_{2}$, respectively.}
        \label{fig:tree_coloring}
    \end{subfigure}
    \hfill
    \begin{subfigure}{0.45\linewidth}
        \centering
        \includegraphics[bb=0 0 112 99, scale=0.6]{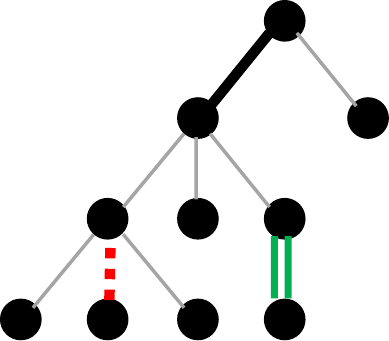}

        \caption{Edge dominating set. Adopted edges are drawn by bold lines and the others thin lines. The colors of bold edges correspond to those in Figure~\ref{fig:tree_coloring}.}
        \label{fig:tree_edge_dominating_set}
    \end{subfigure}
    \caption{Unigraphic coloring and a corresponding edge dominating set.}
    \label{fig:tree_colorings}
\end{figure}

\begin{theorem}
    There exists an algorithm which calculates the unigraph number of a tree in $\bigO{n}$ time.
\end{theorem}
\begin{proof}
    Calculate the edge domination number of the tree in $\bigO{n}$ time~\cite{hedetniemi1977edge}.
    By Lemma~\ref{le:tree_dom}, it is also the unigraph number of the tree. \qed
\end{proof}

Note that the unigraphic coloring obtained by an edge dominating set is not necessarily a strongly unigraphic coloring.
Figure~\ref{fig:caterpillars} shows an anticase.
It is an interesting open problem whether we can calculate the strong unigraph number of a tree in polynomial time.

\begin{figure}[t]
    \centering
    \begin{subfigure}{.45\linewidth}
        \centering
        \includegraphics[bb=0 0 284 49, scale=0.5]{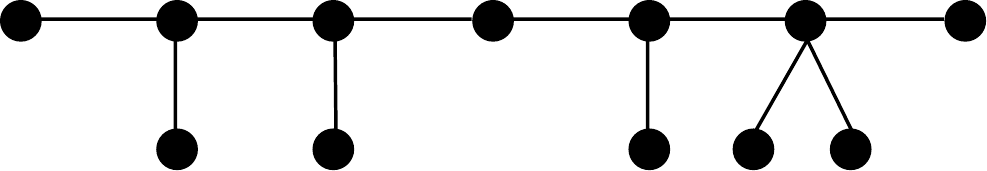}
        \caption{Tree.}
        \label{fig:caterpillar}
    \end{subfigure}
    \hfill
    \begin{subfigure}{.45\linewidth}
        \centering
        \includegraphics[bb=0 0 284 49, scale=0.5]{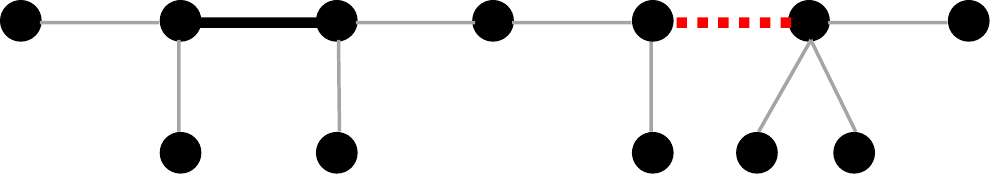}
        \caption{Edge dominating set.}
        \label{fig:caterpillar_edge_dom}
    \end{subfigure}
    \hfill
    \begin{subfigure}{.45\linewidth}
        \centering
        \includegraphics[bb=0 0 284 49, scale=0.5]{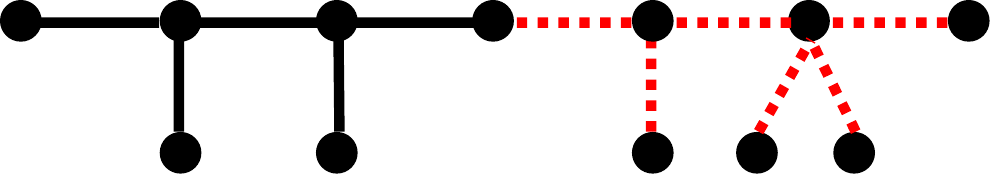}
        \caption{Unigraphic coloring obtained by the edge dominating set in Figure~\ref{fig:caterpillar_edge_dom}.}
        \label{fig:caterpillar_coloring}
    \end{subfigure}
    \hfill
    \begin{subfigure}{.45\linewidth}
        \centering
        \includegraphics[bb=0 0 284 49, scale=0.5]{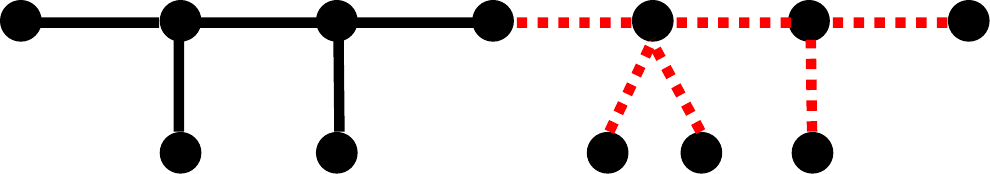}
        \caption{2-colored graph which has the same colored degree sequence as the colored graph in Figure~\ref{fig:caterpillar_coloring} but is not isomorphic to the graph in Figure~\ref{fig:caterpillar}.}
        \label{fig:caterpillar_not_iso}
    \end{subfigure}
    \caption{Unigraphic but not strongly unigraphic coloring of a tree.}
    \label{fig:caterpillars}
\end{figure}

\section{Concluding remarks}\label{sec:conclusions}
In this paper, we have introduced two types of generalizations of unigraphs: \emph{$k$-unigraphs} and \emph{$k$-strong unigraphs}.
We have described the relationship between $k$-(strong) unigraphs and the subgraph isomorphism problem and analyzed the basic properties of $k$-(strong) unigraphs.
This paper opens many research topics.
There are two main viewpoints to the $k$-(strong) unigraphs.

The first viewpoint is the algorithmic one.
When we are given a graph and a positive integer $k$, what is the complexity of the problem to determine whether the (strong) unigraph number of the graph is at most $k$?
The decision problem of the unigraph number is in NP because, when we are given a partition of the edge set of a graph with at most $k$ sets as a certificate, we can check whether each set forms a unigraph using the linear-time recognition algorithm.
However, whether the problem is in NP-complete or not remains open.
Moreover, as for the strong unigraph number, we do not know even whether the decision problem is in NP because, when we are given a colored degree sequence as a certificate, it is unclear that we can determine whether it has a unique realization in polynomial time.
If the calculation of the (strong) unigraph number of a graph is difficult, can we design approximation algorithms or fixed-parameter algorithms?
It is also important to utilize the (strong) unigraph number for other problems.
Is there any problem which can be solved easily for graphs with small (strong) unigraph numbers?
A hopeful example is the subgraph isomorphism problem, which we have mentioned in Section~\ref{sec:intro}.

The second viewpoint is the graph-theoretic one.
How is the (strong) unigraph number of a graph related to other invariants of the graph?
Is the (strong) unigraph number of a graph bounded for some graph classes?

We expect that the concept of $k$-(strong) unigraphs will open many future research topics.

\subsubsection{Acknowledgement}
This work was partly supported by JSPS KAKENHI Grant Numbers JP15H05711, JP18K04610, JP18H04091.

%
%
%
\bibliographystyle{splncs04}

\begin{thebibliography}{10}
\providecommand{\url}[1]{\texttt{#1}}
\providecommand{\urlprefix}{URL }
\providecommand{\doi}[1]{https://doi.org/#1}

\bibitem{barrus2012switch}
Barrus, M.D.: On 2-switches and isomorphism classes. Discrete Mathematics
  \textbf{312}(15),  2217--2222 (2012)

\bibitem{barrus2013hereditary}
Barrus, M.D.: Hereditary unigraphs and {Erd{\"o}s-Gallai} equalities. Discrete
  Mathematics  \textbf{313}(21),  2469--2481 (2013)

\bibitem{borri2011recognition}
Borri, A., Calamoneri, T., Petreschi, R.: Recognition of unigraphs through
  superposition of graphs. Journal of Graph Algorithms and Applications
  \textbf{15},  323--343 (2011)

\bibitem{chvatal1977aggregation}
Chv{\'a}tal, V., Hammer, P.L.: Aggregation of inequalities in integer
  programming. In: Hammer, P., Johnson, E., Korte, B., Nemhauser, G. (eds.)
  Studies in Integer Programming, Annals of Discrete Mathematics, vol.~1, pp.
  145--162. Elsevier (1977)

\bibitem{foldes1978class}
Foldes, S., Hammer, P.: On a class of matroid-producing graphs. In: Colloq.
  Math. Soc. J. Bolyai (Combinatorics). vol.~18, pp. 331--352 (1978)

\bibitem{hammer2004splitoids}
Hammer, P., Zverovich, I.: Splitoids. Graph Theory Notes N. Y.  \textbf{46},
  36--40

\bibitem{hardy2007kterminal}
Hardy, G., Lucet, C., Limnios, N.: K-terminal network reliability measures with
  binary decision diagrams. IEEE Transactions on Reliability  \textbf{56}(3),
  506--515 (2007)

\bibitem{hedetniemi1977edge}
Hedetniemi, S., Mitchell, S.: Edge domination in trees. In: Proc.\ of 8th SE
  Conf. Combin., Graph Theory and Computing, Congr. Numer. vol.~19, pp.
  489--509 (1977)

\bibitem{inoue2014distribution}
Inoue, T., Takano, K., Watanabe, T., Kawahara, J., Yoshinaka, R., Kishimoto,
  A., Tsuda, K., Minato, S., Hayashi, Y.: Distribution loss minimization with
  guaranteed error bound. IEEE Transactions on Smart Grid  \textbf{5}(1),
  102--111 (2014)

\bibitem{johnson1975simple}
Johnson, R.H.: Simple separable graphs. Pacific J. Math.  \textbf{56}(1),
  143--158 (1975)

\bibitem{kawahara2019colorful}
Kawahara, J., Saitoh, T., Suzuki, H., Yoshinaka, R.: Colorful frontier-based
  search: Implicit enumeration of chordal and interval subgraphs. In: Proc.\ of
  Special Event on Analysis of Experimental Algorithms (SEA$^2$ 2019) (2019 (to
  appear))

\bibitem{kleitman1975note}
Kleitman, D.J., Li, S.Y.: A note on unigraphic sequences. Studies in Applied
  Mathematics  \textbf{54}(4),  283--287 (1975)

\bibitem{li1975graphic}
Li, S.Y.R.: Graphic sequences with unique realization. Journal of Combinatorial
  Theory, Series B  \textbf{19}(1),  42--68 (1975)

\bibitem{maehara2017exact}
Maehara, T., Suzuki, H., Ishihata, M.: Exact computation of influence spread by
  binary decision diagrams. In: Proc.\ of the 26th International World Wide
  Conference (WWW 2017). pp. 947--956 (2017)

\bibitem{marchioro1984degree}
Marchioro, P., Morgana, A., Petreschi, R., Simeone, B.: Degree sequences of
  matrogenic graphs. Discrete Mathematics  \textbf{51}(1),  47--61 (1984)

\bibitem{peled1977matroidal}
Peled, U.N.: Matroidal graphs. Discrete Mathematics  \textbf{20},  263--286
  (1977)

\bibitem{sekine1995tutte}
Sekine, K., Imai, H., Tani, S.: Computing the {Tutte} polynomial of a graph of
  moderate size. In: Proc.\ of the 6th International Symposium on Algorithms
  and Computation (ISAAC 1995). pp. 224--233 (1995)

\bibitem{tyshkevich2000decomposition}
Tyshkevich, R.: Decomposition of graphical sequences and unigraphs. Discrete
  Mathematics  \textbf{220}(1),  201--238 (2000)

\bibitem{yannakakis1980edge}
Yannakakis, M., Gavril, F.: Edge dominating sets in graphs. SIAM Journal on
  Applied Mathematics  \textbf{38}(3),  364--372 (1980)

\end{thebibliography}

\appendix
\section{Proof of Lemma~\ref{le:disconnected_unigraph}}
\label{sec:proof_disconnected_unigraph}
    We proof the lemma by contradiction.
    Assume that there are two distinct connected components $C_1$ and $C_2$ in $G$ such that both of them have at least three vertices.
    We choose an edge $e_i$ from $C_i\,(i \in [2])$ as follows:
    (a) If $C_i$ is not a tree, there exists an edge contained in some cycle. Let the edge be $e_i$. In this case, $e_i$ is not a bridge in $C_i$.
    (b) If $C_i$ is a tree, there exists a vertex whose degree is one. Let the (unique) edge incident to the vertex be $e_i$. In this case, $e_i$ is a bridge in $C_i$.
    In the following, we denote $e_i = \edge{u_i}{v_i}\,(i \in [2])$.

    We consider the following two cases and show that we can construct another graph $G'$ which has the same degree set as $G$ but is not isomorphic to $G$.
    If both $e_1$ and $e_2$ are chosen in the case (b), without loss of generality, we assume that $u_1$ and $u_2$ have the degrees one.
    Consider a graph $G'$ obtained from $G$ by deleting edges $e_1$ and $e_2$ and adding edges $e'_1 = \edge{u_1}{u_2}$ and $e'_2 = \edge{v_1}{v_2}$.
    Note that $D(G') = D(G)$.
    $C_1$ and $C_2$ in $G$ change into $C'_1$ and $C'_2$ in $G'$, whose set of vertices are $\set{u_1, u_2}$ and $(V(C_1) \setminus \set{u_1}) \cup (V(C_2) \setminus \set{u_2})$, respectively.
    While $|V(C_1)| \geq 3$ and $|V(C_2)| \geq 3$ hold, $|V(C'_1)| = 2 < 3$ holds.
    Therefore, $G'$ is not isomorphic to $G$.

    Otherwise, if at least one of $e_1$ or $e_2$ are chosen in the case (a), without loss of generality, we assume that $e_1$ is chosen in the case (a).
    In this case, $e_1$ is not a bridge of $C_1$.
    Consider a graph $G'$ obtained in the same way as in the above.
    Since $e_1$ was not a bridge of $C_1$, $C_1$ and $C_2$ in $G$ are merged into one connected component in $G'$.
    Therefore, the number of connected component in $G'$ is one less than that of $G$, meaning that $G'$ is not isomorphic to $G$ although $D(G') = D(G)$.
\end{document}